\documentclass[conference]{IEEEtran}
\IEEEoverridecommandlockouts

\usepackage{cite}
\usepackage{amsmath,amssymb,amsfonts,amsthm}
\usepackage{algorithmic}
\usepackage{graphicx}
\usepackage{textcomp}
\usepackage{xcolor}
\usepackage{import}
\usepackage{color}
\usepackage[shortlabels]{enumitem}  
\usepackage{kbordermatrix}
\usepackage[hidelinks]{hyperref}

\graphicspath{{figures/}}
\newtheorem{thm}{Theorem}[section]
\newtheorem{lemma}[thm]{Lemma}

\theoremstyle{definition}
\newtheorem{definition}{Definition}[section]

\theoremstyle{remark}

\newcommand{\arrayproduct}{\mathbin{{\oplus}.{\otimes}}}
\newcommand{\Eout}{\mathbf{E}_\mathrm{out}}
\newcommand{\Ein}{\mathbf{E}_\mathrm{in}}
\newcommand{\Kout}{V} %{K_\mathrm{out}}
\newcommand{\Kin}{V} %{K_\mathrm{in}}
\newcommand{\vertexset}{V} %{\Kout \cup \Kin}
\DeclareMathOperator{\size}{size}
\newcommand{\nullrowlabel}{\!\!\!\!\!\!}
\newcommand{\tightrowlabel}{\!\!\!\!\!}

\def\BibTeX{{\rm B\kern-.05em{\sc i\kern-.025em b}\kern-.08em
    T\kern-.1667em\lower.7ex\hbox{E}\kern-.125emX}}
\begin{document}

\title{Algebraic Conditions on One-Step Breadth-First Search}

\author{\IEEEauthorblockN{Emma Fu, Hayden Jananthan, Jeremy Kepner}
\IEEEauthorblockA{\textit{Massachusetts Institute of Technology}}
}

\IEEEoverridecommandlockouts
\IEEEpubid{\makebox[\columnwidth]{979-8-3503-0965-2/23/\$31.00~\copyright2023 IEEE \hfill} \hspace{\columnsep}\makebox[\columnwidth]{ }}

\maketitle
\IEEEpubidadjcol

\begin{abstract}
    The GraphBLAS community has demonstrated the power of linear algebra-leveraged graph algorithms, such as matrix-vector products for breadth-first search (BFS) traversals. This paper investigates the algebraic conditions needed for such computations when working with directed hypergraphs, represented by incidence arrays with entries from an arbitrary value set with binary addition and multiplication operations. Our results show the one-step BFS traversal is equivalent to requiring specific algebraic properties of those operations. Assuming identity elements $0, 1$ for operations, we show that the two operations must be zero-sum-free, zero-divisor-free, and $0$ must be an annihilator under multiplication. Additionally, associativity and commutativity are shown to be necessary and sufficient for independence of the one-step BFS computation from several arbitrary conventions. These results aid in application and algorithm development by determining the efficacy of a value set in computations.
\end{abstract}


\section{Introduction}
\label{introduction section}

\let\thefootnote\relax\footnotetext{Research was sponsored by the United States Air Force Research Laboratory and the Department of the Air Force Artificial Intelligence Accelerator and was accomplished under Cooperative Agreement Number FA8750-19-2-1000. The views and conclusions
contained in this document are those of the authors and should not be interpreted as representing the official policies, either expressed or implied, of the Department of the Air Force or the U.S. Government. The U.S. Government is authorized to reproduce and distribute reprints for Government purposes notwithstanding any copyright notation herein.}

With the increasingly large amount of data being processed today, hypergraphs are becoming an extremely relevant tool for data analysis \cite{guo2023lahyper}. They can be used to model many real-world relationships, like social networks, the web, or biological processes \cite{lee2021realworld, klamt2009biology}. The representation of graphs by way of adjacency and incidence matrices is historically well-established \cite{kirchhoff1847ueber,poincare1900second}, but the GraphBLAS \cite{graphblas2022} community has particularly championed the computational and algorithmic power of graph algorithms expressed in the language of linear algebra by way of the use of adjacency and incidence matrices. An additional aspect of GraphBLAS's approach is the support of many different semirings beyond the standard real number and complex number algebras \cite{mattson2013standards,kepner2017graphblas}, structures of the form $(\mathbb{V}, \oplus, \otimes, 0, 1)$ satisfying some additional algebraic conditions. 

A central observation of the GraphBLAS approach to graph algorithms in the language of linear algebra is that vector-matrix multiplication represents a one-step breadth-first search (BFS) traversal of a graph, i.e., computes the frontier vertices from a given subset of vertices \cite{kepner2011graph}. For example, \cite{guo2023lahyper} shows that for undirected hypergraphs,
\begin{equation} \label{linalgbfs}
    \mathbf{v} \mapsto (\mathbf{v} \arrayproduct \Eout^\intercal) \arrayproduct \Ein
\end{equation}
(with $\oplus = \mathbin{\mathrm{OR}}$ and $\otimes = \mathbin{\mathrm{AND}}$ on $\{0, 1\}$) computes a one-step BFS traversal of the hypergraph, where $\Eout = \Ein$ is an incidence matrix representation of the hypergraph. Note the use of incidence matrices---adjacency matrices are insufficient for representing hypergraphs as well as multigraphs (with distinguishable edges) \cite{kepner2015reasons}. 

While \cite{guo2023lahyper} shows \eqref{linalgbfs} (which we term 'LinAlgBFS' from this point onward) is valid when working over the Boolean value set $\{0, 1\}$, it suggests the more general question: under what conditions on the value set $(\mathbb{V}, \oplus, \otimes, 0, 1)$ is LinAlgBFS valid with respect to $(\mathbb{V}, \oplus, \otimes, 0, 1)$, i.e., properly compute a one-step BFS traversal? For the purposes of this paper, we focus solely on directed hypergraphs.

The remainder of \S\ref{introduction section} consists of necessary definitions, notation, and conventions.
\S\ref{linalgbfs validity section} establishes precise algebraic conditions necessary and sufficient for LinAlgBFS to be valid (Theorem~\ref{t2.1}).
\S\ref{convention independence section} establishes necessary and sufficient algebraic conditions (Theorems~\ref{associativity theorem}, \ref{multiplication commutativity theorem}, \& \ref{addition commutativity theorem}) for LinAlgBFS to be independent of several conventions implicit in LinAlgBFS---the ordering of vertices and hyperedges when computing array products, the use of row versus column vectors,
%right- versus left-associativity 
and the computation of $\mathbf{v} \arrayproduct \Eout^\intercal$ versus $\Eout^\intercal \arrayproduct \Ein$ first.

\subsection{Definitions, Notation, \& Conventions}

For the remainder of the paper, $G = (\vertexset, K)$ denotes a directed hypergraph with hyperedge set $K$ and vertex set $\vertexset$. We assume there are fixed total orderings of $V$ and $K$.

We work with associative arrays to better capture the role of vertices and hyperedges in our arrays. The standard reference is \cite{kepner2018bigdata}, though we follow the approach of \cite{jananthan2017construct} in which the only prior assumptions made about $\oplus, \otimes$ (``addition'', ``multiplication'') is that they have identity elements $0, 1 \in \mathbb{V}$, respectively.

\begin{definition}[associative array]
    An \emph{associative array} is a map $\mathbf{A} \colon K_1 \times K_2\rightarrow \mathbb{V}$ where $K_1, K_2$ are finite, totally ordered key sets and $\mathbb{V}$ is the underlying set of the value set $(\mathbb{V}, \oplus, \otimes, 0, 1)$.

    $\mathbf{A}$ is a \emph{row vector} if $\size(K_1) = 1$ and a \emph{column vector} if $\size(K_2) = 1$.
\end{definition}

\begin{definition}[array product]
    Given $\mathbf{A} \colon K_1 \times K_2 \to \mathbb{V}$ and $\mathbf{B} \colon K_2 \times K_3 \to \mathbb{V}$, their \emph{array product} $\mathbf{C} = \mathbf{A} \arrayproduct \mathbf{B} \colon K_1 \times K_3 \to \mathbb{V}$ is defined for $(k_1, k_3) \in K_1 \times K_3$ by
    \begin{equation*}
        \mathbf{C}(k_1, k_3) = \bigoplus_{k_2 \in K_2}{\mathbf{A}(k_1, k_2) \otimes \mathbf{B}(k_2, k_3)}.
    \end{equation*}
\end{definition}

As $\oplus$ is not assumed associative or commutative, the iterated operator $\bigoplus_{k_2 \in K_2}$ is interpreted using right associativity (e.g., $u \oplus v \oplus w = u \oplus (v \oplus w)$) and ordered according to the total ordering of $K_2$.

\begin{definition}[{out-, in-}incidence arrays]
    An associative array $\Eout \colon \Kout \times K \rightarrow \mathbb{V}$ is an \emph{out-incidence array} if for all hyperedges $k$ and vertices $a$, $\Eout(k, a) \neq 0$ if and only if $a$ is an initial vertex of $k$.
    
    An associative array $\Ein \colon \Kin \times K  \rightarrow \mathbb{V}$ is an \emph{in-incidence array} if for all hyperedges $k$ and vertices $a$, $\Ein(k, a) \neq 0$ if and only if $a$ is a terminal vertex of $k$.
\end{definition}

\section{Algebraic Criteria for LinAlgBFS Validity}
\label{linalgbfs validity section}

To further clarify the algebraic requirements for the validity of LinAlgBFS, we break LinAlgBFS into two steps, each of which have graph-theoretic interpretations. Let $\mathbf{e} = \mathbf{v} \arrayproduct \Eout^\intercal$ and $\mathbf{w} = \mathbf{e} \arrayproduct \Ein$. Consider the following statements:
\begin{equation}
\label{eq1}
    \begin{split}
    & \bigoplus\limits_{a\in \Kout}{\mathbf{v}(a) \otimes \Eout^\intercal(a,k)} = 0 ~\text{if and only if}~ \\ 
    & \text{$\Eout^\intercal(a,k) = 0$ or $\mathbf{v}(a) = 0$ for all $(a, k) \in \Kout \times K$}
    \end{split} \tag{$\ast$}
\end{equation}
(I.e., $\mathbf{e} = \mathbf{v} \arrayproduct \Eout^\intercal$ should indicate exactly the hyperedges having an initial vertex indicated in $\mathbf{v}$.)

\begin{equation}
\label{eq2}
    \begin{split}
    & \bigoplus\limits_{k\in K}{\mathbf{e}(k) \otimes \Ein(k,a)} = 0 ~\text{if and only if}~ \\ 
    & \text{$\Ein(k,a) = 0$ or $\mathbf{e}(k) = 0$ for all $(a, k) \in \Kin \times K$}
    \end{split} \tag{$\dagger$}
\end{equation}
(I.e., $\mathbf{w} = \mathbf{e} \arrayproduct \Ein$ should indicate exactly the terminal vertices of the hyperedges indicated in $\mathbf{e}$.)

\begin{thm}
\label{t2.1}
    Given a value set $(\mathbb{V}, \oplus, \otimes, 0, 1)$, the following are equivalent.
    \begin{enumerate}[(i)]
        \item \eqref{eq1} is valid with respect to $(\mathbb{V}, \oplus, \otimes, 0, 1)$ for any directed hypergraph $G$.
        \item \eqref{eq2} is valid with respect to $(\mathbb{V}, \oplus, \otimes, 0, 1)$ for any directed hypergraph $G$.
        \item LinAlgBFS is valid with respect to $(\mathbb{V}, \oplus, \otimes, 0, 1)$ for any directed hypergraph $G$.
        \item $(\mathbb{V}, \oplus, \otimes, 0, 1)$ satisfies the following algebraic conditions. For all $v, w \in \mathbb{V}$:
        \begin{enumerate}[(I)]
            \item Zero-sum-free: $v \oplus w = 0 \implies v = w = 0$.
            \item Zero-divisor free: $v \otimes w = 0 \implies v = 0 \text{ or } w = 0$.
            \item Zero annihilates: $v \otimes 0 = 0 \otimes v = 0$.
        \end{enumerate}
    \end{enumerate}
\end{thm}

\begin{figure} 
    \begin{center} 
    \includegraphics[scale=0.6]{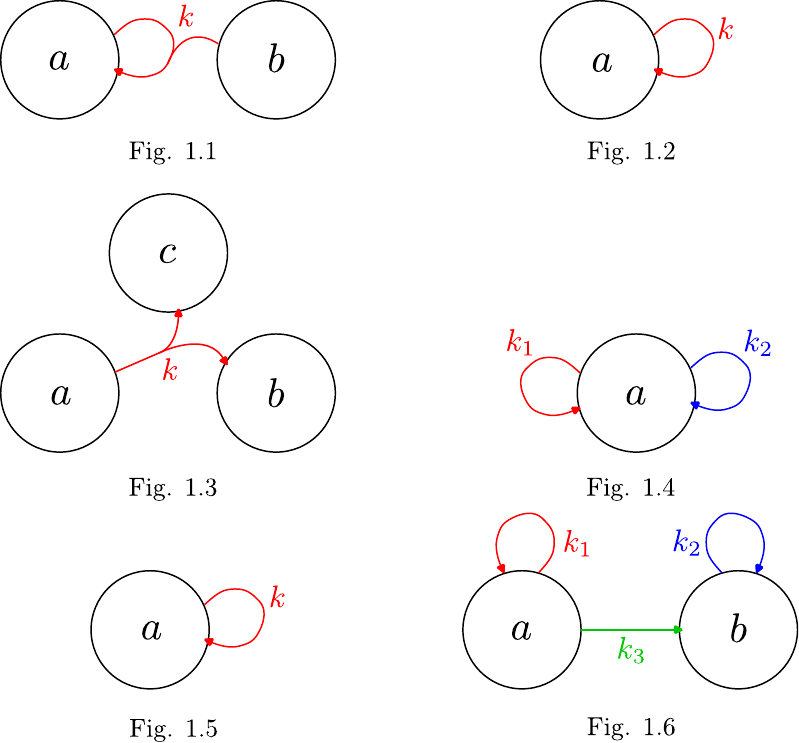}
    \caption{Explicit directed hypergraphs used to prove necessity of zero-sum-freeness, zero-divisor-freeness, and $0$ an annihilator for $\otimes$ for \eqref{eq1} and \eqref{eq2}, respectively.\label{f1}}     
    \end{center} 
\end{figure} 

\begin{proof}
    We start by showing \textit{(i)} implies \textit{(iv)}.

    % (i) -> (iv)
    \begin{lemma} 
        \eqref{eq1} $\implies$ $\mathbb{V}$ is zero-sum-free.
    \label{l2.2}
    \end{lemma}
    \begin{proof}
        Let $G = (\{a, b\}, (\{a, b\}, \{a\}))$ be the directed hypergraph 
        as in Figure~\ref{f1}.1. Given $v, w \in \mathbb{V} \setminus \{0\}$, define 
        \begin{equation*}
             \mathbf{v} = \kbordermatrix{\nullrowlabel & a & b \\ \nullrowlabel & 1 & 1} \quad \text{and} \quad \Eout = \kbordermatrix{\nullrowlabel & a & b \\ k \tightrowlabel & v & w}.
        \end{equation*}
        Then
            $\displaystyle \mathbf{v} \arrayproduct \Eout^\intercal = \kbordermatrix{\nullrowlabel & k \\ \nullrowlabel & (v \otimes 1) \oplus (w \otimes 1) } = \kbordermatrix{\nullrowlabel & k \\ \nullrowlabel & v \oplus w }$.
        
        For \eqref{eq1} to hold in this case, we must have $v \oplus w \neq 0$.
        Zero-sum-freeness requires that $v \oplus w \neq 0$ whenever $v, w$ are not both $0$; we've shown $v \oplus w \neq 0$ when both $v, w$ are nonzero, so it just remains to address the cases where exactly one of $v, w$ are $0$. Indeed, $0 \oplus w = w \neq 0$ and $v \oplus 0 = v \neq 0$ since $0$ is an identity for $\oplus$, so $\mathbb{V}$ is zero-sum-free.
    \end{proof}
    
    \begin{lemma}
    \label{l2.3}
    	\eqref{eq1} $\implies$ $\mathbb{V}$ is zero-divisor-free.
    \end{lemma}
    \begin{proof}
        Let $G$ be the directed hypergraph consisting of one vertex $a$ with a self-looped hyperedge $k$ from $\{a\}$ to $\{a\}$, as in Figure~\ref{f1}.2. Given $v, w \in \mathbb{V} \setminus \{0\}$, define 
        \begin{equation*}
            \mathbf{v} = \kbordermatrix{\nullrowlabel & a \\ \nullrowlabel & v } \quad \text{and} \quad \Eout = \kbordermatrix{\nullrowlabel & a \\ k \tightrowlabel & w }.
        \end{equation*}
        Then 
        $ \displaystyle
            \mathbf{v} \arrayproduct \Eout^\intercal = \kbordermatrix{\nullrowlabel & k \\ \nullrowlabel & v\otimes w}.
        $
        For \eqref{eq1} to hold in this case, we must have $v \otimes w \neq 0$, hence $\mathbb{V}$ is zero-divisor-free.
    \end{proof}
    
    \begin{lemma}
    \label{l2.4}
    	\eqref{eq1} $\implies$ $0$ is annihilator for $\otimes$.
    \end{lemma}
    \begin{proof}
        Let $G$ be the directed hypergraph consisting of three vertices $a, b, c$ with a hyperedge $k$ from $\{a\}$ to $\{b, c\}$ as in Figure~\ref{f1}.3. Given $v\in \mathbb{V}$, define 
        \begin{equation*}
            \mathbf{v} = \kbordermatrix{\nullrowlabel & a & b & c \\ \nullrowlabel & 0 & 0 & v } \quad \text{and} \quad \Eout = \kbordermatrix{\nullrowlabel & a & b & c \\ k \tightrowlabel & v & 0 & 0 }.
        \end{equation*}
        Then 
        $\displaystyle
            \mathbf{v} \arrayproduct \Eout^\intercal = \kbordermatrix{\nullrowlabel & k \\ \nullrowlabel & (0 \otimes v) \oplus (0 \otimes 0) \oplus (v \otimes 0)}.
        $
        \eqref{eq1} implies $(0 \otimes v) \oplus (0 \otimes 0) \oplus (v \otimes 0) = 0$. Since Lemma~\ref{l2.2} establishes that \eqref{eq1} implies $\mathbb{V}$ is zero-sum-free, $0 \otimes v = 0 \otimes 0 = v \otimes 0 = 0$, showing $0$ is an annihilator for $\otimes$. 
    \end{proof}

    The directed hypergraphs used in Lemmas~\ref{l2.2}, \ref{l2.3}, \& \ref{l2.4} to show that \textit{(i)} implies \textit{(iv)} may be similarly used to show that \textit{(iii)} implies \textit{(iv)} using the same assignments for $\mathbf{v}$ and $\Eout$, and setting $\Ein$ to be the in-incidence array where $\Ein(k, a) \neq 0$ implies $\Ein(k, a) = 1$ in each case. 
    
    Next, we show \textit{(ii)} implies \textit{(iv)}.
    
    % (ii) -> (iv):
    \begin{lemma} 
    \label{l2.7}
        \eqref{eq2} $\implies$ $\mathbb{V}$ is zero-sum-free.
    \end{lemma}
    \begin{proof}
        Let $G$ be the directed hypergraph consisting of one vertex $a$ with two self-looped hyperedges $k_1,k_2$, both from $\{a\}$ to $\{a\}$, as in Figure~\ref{f1}.4. Given $v,w \in \mathbb{V}\setminus \{0\}$, define
        \begin{equation*}
            \mathbf{e} = \kbordermatrix{\nullrowlabel & k_1 & k_2 \\ \nullrowlabel & 1 & 1} \quad \text{and} \quad \Ein = \kbordermatrix{\nullrowlabel & a \\ k_1 \tightrowlabel & v \\ k_2 \tightrowlabel & w}.
        \end{equation*}
        Then
        $\displaystyle
            \mathbf{e} \arrayproduct \Ein = \kbordermatrix{\nullrowlabel & a \\ \nullrowlabel & (v \otimes 1) \oplus (w \otimes 1)} = \kbordermatrix{\nullrowlabel & a \\ \nullrowlabel & v \oplus w }.
        $
        For \eqref{eq2} to hold in this case, we must have $v \oplus w \neq 0$, hence $\mathbb{V}$ is zero-sum-free. (The cases where exactly one of $v, w$ are $0$ are addressed in the same way as in the proof of Lemma~\ref{l2.2}.)
    \end{proof}
    
    \begin{lemma}
    \label{l2.8}
    	\eqref{eq2} $\implies$ $\mathbb{V}$ is zero-divisor-free.
    \end{lemma}
    \begin{proof}
        Let $G$ be the directed hypergraph consisting of one vertex $a$ with a self-looped hyperedge $k$ from $\{a\}$ to $\{a\}$, as in Figure~\ref{f1}.5.
        Given $v,w \in \mathbb{V}\setminus\{0\}$, define 
        \begin{equation*}
            \mathbf{e} = \kbordermatrix{\nullrowlabel & k \\ \nullrowlabel & v } \quad \text{and} \quad \Ein = \kbordermatrix{\nullrowlabel & a \\ k \tightrowlabel & w}.
        \end{equation*}
        Then 
        $\displaystyle
            \mathbf{e} \arrayproduct \Eout^\intercal = \kbordermatrix{\nullrowlabel & a \\ \nullrowlabel & v\otimes w }.
        $
        For \eqref{eq1} to hold in this case, we must have $v \otimes w \neq 0$, hence $\mathbb{V}$ is zero-divisor-free.
    \end{proof}
    
    \begin{lemma}
    \label{l2.9}
    	\eqref{eq2} $\implies$ $0$ is annihilator for $\otimes$.
    \end{lemma}
    \begin{proof}
    Let $G$ be the directed hypergraph consisting of two vertices $a, b$ with three hyperedges, a self-loop to and from $\{a\}$, one from $\{a\}$ to $\{b\}$, and a second self-loop to and from $\{b\}$, as in Figure~\ref{f1}.6. Given $v \in \mathbb{V}$, define 
    \begin{equation*}
         \mathbf{e} = \kbordermatrix{\nullrowlabel & k_1 & k_2 & k_3 \\ \nullrowlabel & 0 & v & 0 } \quad \text{and} \quad \Ein = \kbordermatrix{\nullrowlabel & a & b \\ k_1 \tightrowlabel & v & 0 \\ k_2 \tightrowlabel & 0 & 1 \\ k_3 \tightrowlabel & 0 & 1 }.
    \end{equation*}
    Then 
    $\displaystyle
         \mathbf{e} \arrayproduct \Eout^\intercal = \kbordermatrix{\nullrowlabel & a & b \\ \nullrowlabel & (0 \otimes v) \oplus (v \otimes 0) \oplus (0 \otimes 0) & v }.
    $ 
    \eqref{eq2} implies $(0 \otimes v) \oplus (v \otimes 0) \oplus (0 \otimes 0) = 0$. Since Lemma~\ref{l2.7} establishes that \eqref{eq2} implies $\mathbb{V}$ is zero-sum-free, $0 \otimes v = v \otimes 0 = 0 \otimes 0 = 0$, showing $0$ is an annihilator for $\otimes$. 
    \end{proof}
    
    Next, we show \textit{(iv)} implies \textit{(i)} and \textit{(ii)}.
    
    % (iv) -> (i) and (ii)
    \begin{lemma}
    \label{l2.5}
        \textit{(iv)} $\implies$ \eqref{eq1}.
    \end{lemma}
    \begin{proof}
        Assume that $\mathbb{V}$ is zero-sum-free, zero-divisor-free, and that $0$ is an annihilator for $\otimes$. Fix $k \in K$. 

        Suppose $\bigoplus_{a \in \Kout}{\mathbf{v}(a) \otimes \Eout^\intercal(a, k)} = 0$. Zero-sum-freeness shows $\mathbf{v}(a) \otimes \Eout^\intercal(a, k) = 0$ for all $a \in \Kout$, after which zero-divisor-freeness shows that for all $a \in \Kout$, $\mathbf{v}(a) = 0$ or $\Eout^\intercal(a, k) = 0$.

        Conversely, if $\mathbf{v}(a) = 0$ or $\Eout^\intercal(a, k) = 0$ for all $a \in \Kout$, then $0$ being an annihilator for $\otimes$ shows $\mathbf{v}(a) \otimes \Eout^\intercal(a, k) = 0$. Because $0$ is an identity for $\oplus$, 
        \begin{equation*}
            \bigoplus_{a \in \Kout}{\mathbf{v}(a) \otimes \Eout^\intercal(a, k)} = \bigoplus_{a \in \Kout}{0} = 0. \qedhere
        \end{equation*}
    \end{proof}

    Showing \textit{(iv)} implies \textit{(ii)} is analogous.
    
    \begin{lemma}
    \label{l2.10}
        \textit{(iv)} $\implies$ \eqref{eq2}.
    \end{lemma}
    \begin{proof}
        Assume that $\mathbb{V}$ is zero-sum-free, zero-divisor-free, and that $0$ is an annihilator for $\otimes$. Fix $a \in \Kin$.

        Suppose $\bigoplus_{k \in K}{\mathbf{e}(k) \otimes \Ein(k, a)} = 0$. Zero-sum-freeness shows $\mathbf{e}(k) \otimes \Ein(k, 1) = 0$, after which zero-divisor-freeness shows that for all $k \in K$, $\mathbf{e}(k) = 0$ or $\Ein(k, a) = 0$.

        Conversely, if $\mathbf{e}(k) = 0$ or $\Ein(k, a) = 0$ for all $k \in K$, then $0$ being an annihilator for $\otimes$ shows $\mathbf{e}(k) \otimes \Ein(k, a) = 0$. Because $0$ is an identity for $\oplus$,
        \begin{equation*}
            \bigoplus_{k \in K}{\mathbf{e}(k) \otimes \Ein(k, a)} = \bigoplus_{k \in K}{0} = 0. \qedhere
        \end{equation*}
    \end{proof}

    % (i) and (ii) -> (iii)
    Finally, we show \textit{(i)} implies \textit{(iii)}. As \textit{(i)} implies \textit{(iv)}, we may assume $\mathbb{V}$ is zero-sum-free and zero-divisor-free and that $0$ is an annihilator for $\otimes$. As such, $0 = ((\mathbf{v} \arrayproduct \Eout^\intercal) \arrayproduct \Ein)(a) = \bigoplus_{k \in K}{(\mathbf{v} \arrayproduct \Eout^\intercal)(k) \otimes \Ein(k, a)}$ implies $(\mathbf{v} \arrayproduct \Eout^\intercal)(k) = 0$ or $\Ein(k, a) = 0$ for every $k \in K$. \textit{(i)} is the validity of \eqref{eq1}, so for a given $k \in K$, $(\mathbf{v} \arrayproduct \Eout^\intercal)(k) = 0$ is equivalent to saying no vertex indicated in $\mathbf{v}$ is an initial vertex of $k$. As such, ``$(\mathbf{v} \arrayproduct \Eout^\intercal)(k) = 0$ or $\Ein(k, a) = 0$ for every $k \in K$'' is equivalent to saying no hyperedge connects a vertex indicated in $\mathbf{v}$ to $a$, as desired.
\end{proof}

\section{Algebraic Conditions for Convention Independence}
\label{convention independence section}

Now that we have established the algebraic conditions under which $(\mathbf{v}\arrayproduct \Eout^\intercal)\arrayproduct \Ein$ correctly computes a one-step BFS traversal, the independence of the latter graph-theoretic interpretation from the conventions mentioned in \S\ref{introduction section}--- the ordering of vertices and hyperedges when computing array products, the use of row versus column vectors, and the computation of $\mathbf{v} \arrayproduct \Eout^\intercal$ versus $\Eout^\intercal \arrayproduct \Ein$ first---suggests that the computation $\mathbf{v} \mapsto (\mathbf{v} \arrayproduct \Eout^\intercal) \arrayproduct \Ein$ should similarly be independent of those conventions.

\subsection{Associativity \& Commutativity of \texorpdfstring{$\oplus$}{Addition}}

Among the most arbitrary of the aforementioned conventions is the assumption of an implicit, fixed total ordering of the vertices and hyperedges in a given directed hypergraph $G$. Independence of such orderings relates to the associativity and commutativity of $\oplus$.

\begin{thm}
\label{addition commutativity theorem}
     Assume $(\mathbb{V}, \oplus, \otimes, 0, 1)$ satisfies the conditions of Theorem~\ref{t2.1}. The following are equivalent.
    \begin{enumerate}[(i)] 
		\item Any LinAlgBFS $(\mathbf{v} \mathbin{{\oplus}.{\otimes}} \mathbf{E}_\mathrm{out}) \mathbin{{\oplus}.{\otimes}} \mathbf{E}_\mathrm{in}^\intercal$ computation is invariant under reordering of vertices. 
		\item $\oplus$ is associative and commutative.
	\end{enumerate} 
\end{thm}
\begin{proof}
    First assume \textit{(i)} holds and let $G$ be the hypergraph consisting of two vertices $a, b$ and one hyperedge $k$ from $\{a, b\}$ to $\{a\}$ as shown in Figure~\ref{f2}.1. Given $u, v \in \mathbb{V}$, define 
    \begin{equation*}
        \mathbf{v} = \kbordermatrix{\nullrowlabel & a & b \\ \nullrowlabel & u & v }, \quad \Eout = \kbordermatrix{\nullrowlabel & a & b \\ k \tightrowlabel & 1 & 1 }, \quad \Ein = \kbordermatrix{\nullrowlabel & a & b \\ k \tightrowlabel & 1 & 0 }.
    \end{equation*}
    The orderings $a < b$ and $b < a$ respectively give 
    \begin{align*}
        \kbordermatrix{\nullrowlabel & a & b \\ \nullrowlabel & u \oplus v & 0 } = (\mathbf{v} \arrayproduct \Eout^\intercal) \arrayproduct \Ein = \kbordermatrix{\nullrowlabel & b & a \\ \nullrowlabel & 0 & v \oplus u }.
    \end{align*}
    By hypothesis, we must have $u \oplus v = v \oplus u$, showing $\oplus$ is commutative. 

    For associativity, let $G$ be the hypergraph with three vertices $a, b, c$ and a hyperedge from $\{a, b, c\}$ to $\{a\}$ as shown in Figure~\ref{f2}.2. Given $u, v, w \in \mathbb{V}$, define
    \begin{equation*}
    \begin{split}
        \mathbf{v} = \kbordermatrix{\nullrowlabel & a & b & c \\ \nullrowlabel & u & v & w}, \quad
        \Eout = \kbordermatrix{\nullrowlabel & a & b & c \\ k \tightrowlabel & 1 & 1 & 1 }, \\
        \text{and}~\Ein = \kbordermatrix{\nullrowlabel & a & b & c \\ k \tightrowlabel & 1 & 0 & 0 }.
    \end{split}
    \end{equation*}
    The orderings $a < b < c$ and $c < a < b$ respectively give
    \begin{align*}
        \kbordermatrix{\nullrowlabel & a & b & c \\ \nullrowlabel & u \oplus (v \oplus w) & 0 & 0 } & = (\mathbf{v} \arrayproduct \Eout^\intercal)\arrayproduct \Ein \\
        & = \kbordermatrix{\nullrowlabel & c & a & b \\ \nullrowlabel & 0 & w \oplus (u \oplus v) & 0 }.
    \end{align*}
    By hypothesis, we must have $u \oplus (v \oplus w) = w \oplus (u \oplus v)$. Coupled with the commutativity of $\oplus$, we have $u \oplus (v \oplus w) = w \oplus (u \oplus v) = (u \oplus v) \oplus w$, proving associativity of $\oplus$.

    It is well-known that associativity and commutativity of the binary operation $\oplus$ implies the iterated operation $\bigoplus$ is invariant under permutation of its arguments, showing \textit{(ii)} implies \textit{(i)}.
\end{proof}

Note that another way associativity of $\oplus$ could be established is requiring that the computation of array products be independent of the choice to use right- versus left-associativity. 

\subsection{Commutativity of \texorpdfstring{$\otimes$}{Multiplication}}

There are competing conventions for LinAlgBFS, one using multiplication on the left by a row vector (the convention used in this paper) and one using multiplication on the right by a column vector. Switching between these two conventions amounts to transposition, and this observation allows for characterizing the independence of LinAlgBFS from the aforementioned convention:

\begin{thm}
\label{multiplication commutativity theorem}
    Assume $(\mathbb{V}, \oplus, \otimes, 0, 1)$ satisfies the conditions of Theorem~\ref{t2.1}. The following are equivalent.
	\begin{enumerate}[(i)] 
		\item $(\mathbf{v} \arrayproduct \Eout^\intercal) \arrayproduct \Ein = (\Ein^\intercal \arrayproduct (\Eout \arrayproduct \mathbf{v}^\intercal))^\intercal$ for any directed hypergraph $G$, any incidence arrays $\mathbf{E}_\mathrm{out}, \mathbf{E}_\mathrm{in}$, and vector $\mathbf{v}$. 
		\item $\otimes$ is commutative.
	\end{enumerate}
\end{thm}
\begin{proof}
    Let $G$ be the directed hypergraph consisting of two vertices $a, b$ and one hyperedge $k$ from $\{a, b\}$ to $\{a\}$ as shown in Figure~\ref{f2}.1. Given $x,y\in \mathbb{V}$, define %$\Ein(k, v_1)=\Ein(k, v_2)=1, \Eout(k,v_1)=x$, and $\mathbf{v}(v_1)=y$, for $x,y\in \mathbb{V}$. 
    \begin{equation*}
        \begin{split}
            \mathbf{v} = \kbordermatrix{\nullrowlabel & a & b \\ \nullrowlabel & v & 0 }, \quad \Eout = \kbordermatrix{\nullrowlabel & a & b \\ k \tightrowlabel & u & 0 }, \quad
            \Ein = \kbordermatrix{\nullrowlabel & a & b \\ k \tightrowlabel & 1 & 1 }.
        \end{split}
    \end{equation*}

    Assume $(\mathbf{v} \arrayproduct \Eout^\intercal) \arrayproduct \Ein = (\Ein^\intercal \arrayproduct (\Eout \arrayproduct \mathbf{v}^\intercal))^\intercal$. Using the values set for $\mathbf{v}, \Eout, \Ein$, we get $x \otimes y = y \otimes x$, i.e., $\otimes$ is commutative.
\end{proof}

\subsection{Associativity of \texorpdfstring{$\otimes$}{Multiplication}}

Prior work shows that for directed graphs, $\mathbf{A} = \Eout^\intercal \arrayproduct \Ein$ computes the adjacency array $\mathbf{A}$ of the directed graph from its incidence arrays $\Eout, \Ein$ under the same algebraic conditions as in Theorem~\ref{t2.1} \cite{jananthan2017construct}. Under those same conditions $\mathbf{v} \mapsto \mathbf{v} \arrayproduct \mathbf{A}$ computes a one-step BFS traversal \cite{shah2011some}. Although adjacency arrays are insufficient for uniquely representing hypergraphs, it stands to reason that $\mathbf{v} \mapsto \mathbf{v} \arrayproduct (\Eout^\intercal \arrayproduct \Ein)$ should give an alternate implementation of LinAlgBFS.

\begin{thm}
\label{associativity theorem}
    Assume $(\mathbb{V}, \oplus, \otimes, 0, 1)$ satisfies conditions of Theorems~\ref{t2.1} \& \ref{addition commutativity theorem}. The following are equivalent.
    \begin{enumerate}[(i)] 
		\item $(\mathbf{v} \arrayproduct \Eout^\intercal) \arrayproduct \Ein = \mathbf{v} \arrayproduct (\Eout^\intercal \arrayproduct \Ein)$ for any directed hypergraph $G$, any incidence arrays $\mathbf{E}_\mathrm{out}, \mathbf{E}_\mathrm{in}$, and vector $\mathbf{v}$. 
        \item $\otimes$ is associative. 
	\end{enumerate} 
\end{thm}
\begin{proof}
    Assume \textit{(i)} holds and let $G$ be the directed hypergraph with two vertices $a, b$, and a single hyperedge $k$ from $\{a, b\}$ to $\{a\}$ as shown in Figure~\ref{f2}.1. Given $u, v, w \in \mathbb{V} \setminus \{0\}$, define
    \begin{equation*}
        \begin{split}
            \mathbf{v} = \kbordermatrix{\nullrowlabel & a & b \\ \nullrowlabel & u & 0 }, \quad \Eout = \kbordermatrix{\nullrowlabel & a & b \\ k \tightrowlabel & v & 1 }, 
            \quad \Ein = \kbordermatrix{\nullrowlabel & a & b \\ k \tightrowlabel & w & 0 }.
        \end{split}
    \end{equation*}
    Evaluating $(\mathbf{v} \arrayproduct \Eout^\intercal) \arrayproduct \Ein$ and $\mathbf{v} \arrayproduct (\Eout^\intercal \arrayproduct \Ein)$ shows
    \begin{equation*}
        \kbordermatrix{\nullrowlabel & a & b \\ \nullrowlabel & (u \otimes v) \otimes w & 0 } = \kbordermatrix{\nullrowlabel & a & b \\ \nullrowlabel & u \otimes (v \otimes w) & 0 },
    \end{equation*}
    hence $(u \otimes v) \otimes w = u \otimes (v \otimes w)$. The cases in which $u$, $v$, or $w$ are $0$ follow from the hypothesis that $0$ is an annihilator for $\otimes$ and hence both expressions $(u \otimes v) \otimes w$ and $u \otimes (v \otimes w)$ evaluate to $0$ in those cases. This shows $\otimes$ is associative.

    The converse follows from the general observation that when $\oplus$ is both associative and commutative and $\otimes$ is associative, then $\mathbf{A} \arrayproduct (\mathbf{B} \arrayproduct \mathbf{C}) = (\mathbf{A} \arrayproduct \mathbf{B}) \arrayproduct \mathbf{C}$ for any associative arrays $\mathbf{A}, \mathbf{B}, \mathbf{C}$ \cite{kepner2018bigdata}.
\end{proof}

Worth noting is that without the prior hypotheses that $(\mathbb{V}, \oplus, \otimes, 0, 1)$ satisfy the conditions of Theorem~\ref{addition commutativity theorem}, condition \textit{(i)} in Theorem~\ref{associativity theorem} implies $\oplus$ is associative by considering the directed hypergraph $G = (\{a, b, c\}, \{(\{a, b\}, \{a, c\}), (\{c\}, \{a, b\})\})$ as shown in Figure~\ref{f2}.3 and the assignments
\begin{equation*}
    \begin{split}
        \mathbf{v} = \kbordermatrix{\nullrowlabel & a & b & c \\ \nullrowlabel & u & v & w }, \quad  \Eout = \kbordermatrix{\nullrowlabel & a & b & c \\ k_1 \tightrowlabel & 1 & 1 & 0 \\ k_2 \tightrowlabel & 0 & 0 & 1}, \\ \text{and}~ \Ein = \kbordermatrix{\nullrowlabel & a & b & c \\ k_1 \tightrowlabel & 1 & 0 & 1 \\ k_2 \tightrowlabel & 1 & 1 & 0 }.
    \end{split} 
\end{equation*}
Evaluating $(\mathbf{v} \arrayproduct \Eout^\intercal) \arrayproduct \Ein$ and $\mathbf{v} \arrayproduct (\Eout^\intercal \arrayproduct \Ein)$ shows
\begin{equation*}
    \kbordermatrix{\nullrowlabel & a & b & c \\ \nullrowlabel & (u \oplus v) \oplus w & w & u \oplus v } = \kbordermatrix{\nullrowlabel & a & b & c \\ \nullrowlabel & u \oplus (v \oplus w) & w & u \oplus v },
\end{equation*}
hence $(u \oplus v) \oplus w = u \oplus (v \oplus w)$.

\begin{figure} 
    \begin{center} 
    \includegraphics[scale=0.6]{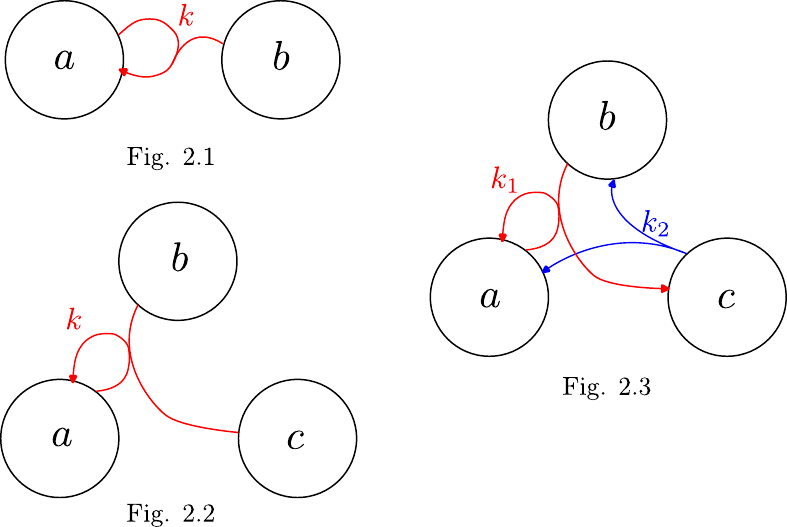}
    \caption{Explicit directed hypergraphs used to prove necessity of associativity and commutativity of $\oplus$ and $\otimes$ given independence from corresponding convention.\label{f2}} 
    \end{center} 
\end{figure}

\section{Conclusions and Future Work}

Based on the high computational performance of graph algorithms written in the language of linear algebra and the rising recognition of the importance of hypergraphs in application, more work needs to be done to better understand what algebraic properties are needed to support those graphs. 

This paper establishes algebraic conditions on the value set $(\mathbb{V}, \oplus, \otimes, 0, 1)$ under the computations $(\mathbf{v} \arrayproduct \Eout^\intercal) \arrayproduct \Ein$ correctly computes one step of a breadth-first search traversal, namely that $\mathbb{V}$ is zero-sum-free and zero-divisor-free, and $0$ is an annihilator for $\otimes$. 

In this paper, the primary focus was on directed hypergraphs. It can be shown that the algebraic requirements for directed hypergraphs---zero-sum-free, zero-divisor-free, and $0$ an annihilator for $\otimes$---apply to undirected hypergraphs as well. However, equivalences in the cases of associativity and commutativity of $\oplus$ and $\otimes$ made more essential use of both the ``directed'' and ``hyper'' aspects of the directed hypergraphs. A potential continuation of this project would be to look other types of graphs, e.g., undirected hypergraphs or directed and undirected multigraphs, and perform a similar analysis to determine which algebraic properties are specific to only a certain type of graph. This allows us to better understand the underlying algebras behind these graphs when using the linear algebraic method to compute graph algorithms, and may provide insight on what we need to investigate to further refine these algorithms.

\section*{Acknowledgements}

The authors wish to acknowledge the following individuals for their contributions and support: W. Arcand, W. Bergeron, D. Bestor, C. Birardi, B. Bond, S. Buckley, C. Byun, G. Floyd, V. Gadepally, D. Gupta, M. Houle, M. Hubbell, M. Jones, A. Klien, C. Leiserson, K. Malvey, P. Michaleas, C. Milner, S. Mohindra, L. Milechin, J. Mullen, R. Patel, S. Pentland, C. Prothmann, A. Prout, A. Reuther, A. Rosa , J. Rountree, D. Rus, M. Sherman, C. Yee.

\bibliographystyle{IEEEtran}
\bibliography{references}

\end{document}